\documentclass[reqno,10pt]{amsart}
\usepackage[off]{auto-pst-pdf}

\usepackage{amssymb}
\usepackage{amsmath}
\usepackage{amsthm}
\usepackage{mathrsfs}
\usepackage{pgf}
\usepackage{color}
\usepackage{graphicx}   

\usepackage{amsmath}
  \usepackage{paralist}
\usepackage{graphicx}  
\usepackage{psfrag}
\usepackage{comment}
\usepackage{esvect}

\newtheorem{theorem}{Theorem}[section]

\newtheorem{lemma}[theorem]{Lemma}
\newtheorem{proposition}[theorem]{Proposition}

\theoremstyle{definition}

\newcommand{\Rz}{\mathbb{R}}
\newcommand{\Nz}{\mathbb{N}}
\newcommand{\Zz}{\mathbb{Z}}

\newcommand{\eps}{\varepsilon}

\newcommand{\PPP}{\color{black}} 
\newcommand{\BBB}{\color{black}}

\newcommand{\EEE}{\color{black}} 
\newcommand{\RRR}{\color{black}} 
 
\newcommand{\MMM}{\color{black}}
\newcommand{\ZZZ}{\color{black}} 
\newcommand{\olive}{\color{black}}

\usepackage{latexsym,amsfonts}
\usepackage{mathrsfs}
\usepackage{centernot}
\usepackage{paralist}

\usepackage{eso-pic}  
\usepackage{pifont}
\usepackage{fixmath} 

\usepackage{metalogo}
\usepackage{booktabs}
\makeatletter

\usepackage{a4wide}

\title[]{Atomistic potentials and the Cauchy-Born rule for carbon nanotubes: a review}
\author{Manuel Friedrich} \author{Edoardo Mainini} \author{Paolo
  Piovano} 


\address[Manuel Friedrich]{Applied Mathematics M\"unster, University of M\"unster\\
Einsteinstrasse 62, 48149 M\"unster, Germany}
\email{manuel.friedrich@uni-muenster.de}
\urladdr{https://www.uni-muenster.de/AMM/Friedrich/index.shtml}

\address[Edoardo Mainini]{Dipartimento di
   Ingegneria meccanica, energetica, gestionale e dei trasporti
   (DIME), Universit\`a degli Studi di Genova,  Via all'Opera Pia 15, I-16145 Genova, Italy 		}
\email{mainini@dime.unige.it}
\urladdr{http://www.dime.unige.it/it/users/edoardo-mainini}

\address[Paolo Piovano]{Faculty of Mathematics,  University of Vienna, Oskar-Morgenstern-Platz 1, A-1090 Vienna, Austria}
\email{paolo.piovano@univie.ac.at}
\urladdr{http://www.mat.univie.ac.at/~piovano/Paolo\texttt{Samp\_Dist\_Corr}Piovano.html}

\begin{document}

\maketitle


\begin{abstract}
Carbon nanotubes are modeled as point particle configurations in the framework of Molecular Mechanics, where interactions are described by means of short range attractive-repulsive potentials. The identification of \olive local energy minimizers \EEE  yields a variational description \olive for the stability of rolled-up hexagonal-lattice  structures. \EEE  Optimality of periodic configurations is preserved under moderate
tension, \PPP  \EEE hence justifying the elastic behavior of carbon nanotubes in the axial traction regime. 
\end{abstract}



\section{Introduction}


Carbon nanotubes are long cylindrical structures  
made of atom-thick layers of
carbon atoms forming {\it $sp^2$ covalent} bonds  \olive and locally \EEE arranging themselves \olive in hexagonal-lattice \EEE patterns
\cite{Clayden12,Dresselhaus-et-al95}.  Among the reasons of the central role \olive that carbon nanotubes have conquered \EEE in several modern technology applications, \olive we mention their \EEE exceptional mechanical properties and tensile strength  \cite{Arroyo05,KDEYT, TEG}, \olive and the fact that \EEE they can grow in length up to $10^8$ times the diameter \cite{WLJ, ZZW}.
\EEE

The
mechanical response  of nanotubes under {\it stretching}  is therefore a topical research subject from the
theoretical \cite{Bajaj13, Favata14, Ru01, Ruoff, YA, Zhang08}, the computational
\cite{Agrawal, Cao07, Han14, Jindal}, and the experimental point of view
\cite{Demczyk, KDEYT,LCS,Warner11, YFAR, YU}. \PPP 
\MMM
Ab initio  models from quantum mechanics
provide an accurate description of the atomic scale    and may  describe characterizing features of carbon nanotube geometry
and mechanics  \cite{Li07,Rochefort99,Yakobson96}. The drawback of these methods is  the rapid
increase in 
computational complexity. Hence, for large systems, 
it is suitable to consider a description in the 
    framework of  molecular
mechanics   \cite{Allinger,Lewars,Rappe}, thus regarding atoms as classical interacting point particles.

Following the latter  approach, we  identify carbon \EEE
nanotubes 
with configurations $\{x_1,
\dots, x_n\}\in \Rz^{3n}$ corresponding to atomic
positions.  The atoms  \ZZZ interact \EEE via  {\it configurational
energies} $E=E(x_1,
\dots, x_n)$ \olive that \ZZZ  depend \EEE on the mutual positions \ZZZ of  particles and \olive \ZZZ are \EEE consistent with the phenomenology of \EEE   $sp^2$ covalent bonding. \EEE
The energy features  classical potentials taking 
into account both attractive-repulsive {\it two-body} effects,
minimized at  a certain  atomic distance, and {\it three-body} terms
favoring specific angles between bonds
\cite{Brenner90,Stillinger-Weber85,Tersoff, T2,T3}.  The $sp^2$-type covalent bonding
implies that each atom has exactly three first neighbors
 and  that bond angles of
$2\pi/3$ are energetically preferred
\cite{Clayden12}. \olive Similar potentials have been previously employed to describe also other carbon allotropes, e.g., graphene in \cite{DPS,emergence, Mainini-Stefanelli12} and fullerene $C_{60}$ in \cite{FPS}. \EEE


In this paper we review the analytical results from \olive \cite{FMPS,MMPS,MMPS-new}: We \EEE
  \olive study \EEE the local minimality of 
periodic configurations, both in the unstreched case and under the effect of small axial
stretching.   More specifically, we prove that, by applying a small
stretching to a nanotube, the energy $E$ is locally strictly minimized by a  specific  periodic
configuration where all atoms see the same local configuration
\olive  (see \EEE Theorem \ref{th: main3} below). Local
minimality is here \olive assessed  \EEE with respect to {\it all}  small  perturbations in
$\Rz^{3n}$, namely not restricting {\it a priori} to \olive only \EEE periodic
perturbations. On the contrary, periodicity is \olive shown  \EEE to emerge as  a byproduct of our variational analysis.
This results provides new motivation for the continuum mechanics approaches to carbon \ZZZ nanotubes and opens \EEE the possibility of developing a discrete-to-continuum theory. \olive In this regard, we refer the reader to the various 
  continuum models provided in \EEE the literature such as
\cite{Arroyo05,Bajaj13,Favata12,Favata15,Poncharal99, Ru01,Wang05}. 

Our result can be seen as a  validation of the so-called {\it
  Cauchy-Born rule} for carbon nanotubes: \olive By \EEE imposing a small tension, the
periodicity cell deforms correspondingly so that the deformation at the microscale is uniformly distributed along the sample.
In other words,  atoms follow the macroscopic deformation and microscopic periodicity is preserved under \ZZZ imposing \EEE a macroscopic stretching.
We stress that such periodicity is \PPP invariably \EEE {\it assumed} in a number of different
contributions, see \cite{Bajaj13,Favata14,Han14,Zhang08} among others, and then exploited in order to compute
tensile strength as well as stretched geometries. Here again our
results provide a theoretical justification of such approaches.

{
Let us mention that, \ZZZ  in principle, \EEE the Cauchy-Born rule would require to check the elastic behavior of the structure with respect to any imposed small displacement.
However, we limit ourselves to consider the most natural tensile stress experiment for the nanotube, i.e., the uniaxial traction. In this setting, we prove that in the small deformation regime, there exists a local minimizer of the energy in which \olive the microscopic periodical pattern \EEE is preserved.  \olive This seems to be the strongest result one can aim for in our context, as {\it global minimality} is not expected in general  for nanotubes and could be instead achieved by completely different \ZZZ topologies. \olive For example, under not too restrictive  hypotheses on \ZZZ the three-body terms,  structures corresponding to the $sp^3$ hybridization \MMM such as diamonds (which  display four nearest neighbors for each bulk atom, arranged at the vertices of a tetrahedron) are energetically favored    (see \cite{Mainini-Stefanelli12} for more details in this direction). \EEE

Finally, we notice that various technical difficulties in the proof originate exactly from the  \EEE fact that our configurational energy is tailored \olive to describe \EEE covalent $sp^2$ bonds.  Indeed, in principle, such modeling assumptions are meant to describe the \olive local minimality of planar  configurations, such as graphene \cite{Mainini-Stefanelli12} and graphene nanoflakes \cite{DPS}. \EEE   When introducing rolled-up structures, however, the \olive  non-planarity \EEE effect destroys the local optimality around each point. Only  a very careful study of the global geometric constraint given by the rolling-up allows to still identify \olive specific nanotube configurations as \EEE local minimizers. \EEE

} 

The paper is organized as \olive follows:  \EEE In Section \ref{Fsection} we introduce a precise geometric description of carbon nanotubes along with the definition of the configurational energy. In Section 3 we state the main result. In Section 4 we collect the main auxiliary tools  and we provide a summary of the proof of the main result.

\section{Geometric and variational description of carbon nanotubes}\label{Fsection}

\subsection{Geometry of zigzag nanotubes}
For \ZZZ a \EEE precise geometric description, we  \ZZZ start \EEE from some distinguished \olive and \EEE highly symmetric configurations \olive related to {\it zigzag} nanotubes. \EEE
The term  zigzag refers to \olive a particular  nanotube geometry depicted in  Figure \ref{Figure1} that ideally \EEE results from  \olive rolling up graphene sheets along \EEE \olive specific axis directions suitably chosen \EEE with respect to the lattice vectors of the planar honeycomb graph \olive of the graphene sheet, \EEE see Figure \ref{fig:nanotube}. 
\olive However, we notice that the definitions \ZZZ introduced here \olive can be easily adapted to other classical choices, \EEE \ZZZ namely \EEE the \olive {\it armchair} \EEE geometry, where the  axis \ZZZ of the cylinder \EEE is instead orthogonal to the bisector of the planar lattice vectors. We refer to
 \cite{MMPS-new} for a precise geometric description in the armchair case.

 \begin{figure}[h]
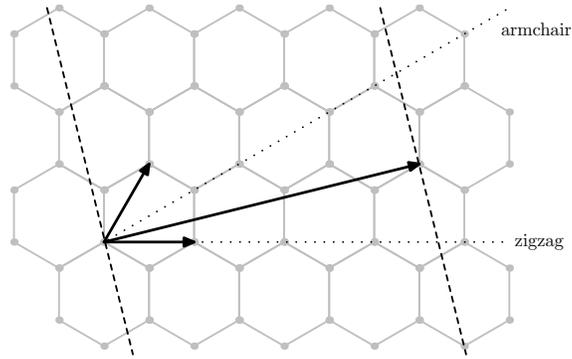

    \pgfdeclareimage[width=0.5\textwidth]{nanotube}{nanotube.pdf}
    \pgfuseimage{nanotube}
    \caption{\PPP Rolling-up \olive of \EEE a graphene sheet:
    a zigzag nanotube is the result of rolling up in such a way that the axis of the cylinder is orthogonal to one of the lattice vectors.
     }
    \label{fig:nanotube}
  \end{figure}

\begin{figure}[h]
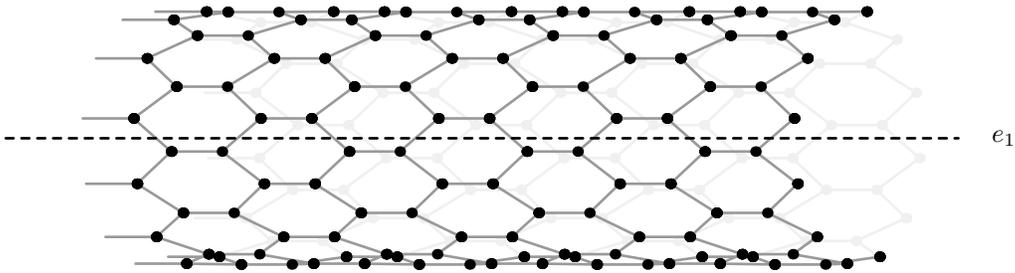

    \pgfdeclareimage[width=0.9\textwidth]{tube}{tube.pdf}
    \pgfuseimage{tube}
    \caption{Zigzag nanotube.}
    \label{Figure1}
  \end{figure}

We let $\ell \in \Nz$, $\ell >3$, and  define the \olive family  $\olive \mathscr{F}\EEE$ of \olive zigzag nanotubes  \EEE 
as the collection of all configurations that, up to isometries, coincide with the set of points
 \begin{align}
&\left(k (\lambda_1+\sigma) +  j (2\sigma+2\lambda_1)  +
  l(2\sigma + \lambda_1), \;\rho \cos \left(\frac{\pi(2i+k)}{\ell}\right),\; \rho\sin
    \left(\frac{\pi(2i+k)}{\ell}\right) \right),
 \ 
  \label{zigzagfamilydefinition}
\end{align}
$  i=1,\dots,\ell, \ j \in \Zz, \ k,l \in \lbrace 0,1 \rbrace$,
for some choice of 
$$\lambda_1 \in (0,\mu/2),\ \ \  \lambda_2 \in (0,\mu/2),\ \ \ \sigma \in (0,\mu/2), \ \ \   \text{and}  \ \ \ \rho\in \left(0,\,\frac{\mu}{4\sin(\pi/(2\ell))}\right)$$
 satisfying the constraints
\begin{align}\label{eq: basic constraints}
2\sigma + 2\lambda_1 = \mu, \ \ \ \ \ \ \ \ \ \sigma^2 +4\rho^2\sin^2\left(\frac{\pi}{2\ell}\right) =\lambda_2^2.
\end{align}

\ZZZ The parameter $\rho$ indicates the diameter of the tube and $\lambda_1$, $\lambda_2$ are the two possibly different lengths of the
covalent bonds in each hexagon of the tube, where the bonds of length $\lambda_1$ are oriented in the $\mathbf e_1$ direction  (see Figure \ref{extra}). The parameter $\mu$ represents instead the minimal period. \EEE \olive We stress that, in the above definition, the only parameter that is fixed a-priori is the number $\ell$ of atoms per section, and that  by \eqref{eq: basic constraints} we can see $\mathscr{F}$ as a three-parameter family depending on $\lambda_1,\lambda_2$ and $\mu$.  In the following we denote a generic configuration in $\mathscr{F}$ by $\mathcal{F}$ and, when we want to refer to the particular parameters $\lambda_1$, $\lambda_2$, and $\mu$ that uniquely determine such configuration we use the notation $\mathcal{F}_{\lambda_1,\lambda_2,\mu}$. Furthermore, the one-parameter subfamily of all configurations characterized by unitary bonds, namely the collection of all configurations $\mathcal{F}_{1,1,\mu}\in\mathscr{F}$ for admissible $\mu$, will be denoted   by $\mathscr F_1$. 
 \EEE

\olive We observe  that  any configuration $\mathcal{F}\in\mathscr{F}$ is periodic in the axis direction, and \EEE it is not difficult  to check that \olive it \EEE enjoys the following properties:
Atoms in $\mathcal{F}$ \olive both \EEE lie on the surface of a cylinder with radius $\rho$ (we denote the direction of the axis of the cylinder  by  $\mathbf e_1$)\olive,
and are \EEE arranged in planar {sections}, perpendicular to $\mathbf e_1$, obtained by fixing $j$, $k$, and $l$ in \eqref{zigzagfamilydefinition}. Each of the sections \PPP contains  exactly \EEE $\ell$ atoms, arranged at the vertices of a regular $\ell$-gon. For each section, the two closest sections are at distance $\sigma$ and $\lambda_1$, respectively.

We shall impose restrictions to the range of the parameters in order to avoid degeneracy of the structure. In particular, we shall impose that each atom has exactly three and only three neighbors at distance less than a reference value chosen as $1.1$ \olive in accordance to the fact that only atoms closer than $1.1$ contribute to the energy that we define  later on in Subsection \ref{energy_def}. \EEE This can be ensured by 
 some simple trigonometry which yields  suitable explicit bounds on the parameters. We do not give detail about this, apart from noticing that the bond lengths $\lambda_1$ and $\lambda_2$ will lie in a \ZZZ neighborhood \EEE of the reference value $1$, the parameter $\mu$ will lie in a neighborhood of $3$, and the bond angles will lie in a neighborhood of $2\pi/3$, which are the values of the perfect unit planar haxagonal lattice (to which we are \olive locally close for very \EEE large $\ell$).  
 For instance, $\lambda_1,\lambda_2\in (0.9,1.1)$ and $\mu\in(2.6,3.1)$ are fine choices, see \cite{FMPS}.


If we denote by $\alpha$ the two angles that are adjacent to a bond in the axis direction and by $\beta$ the remaining angles (see also Figure \ref{extra}), we have the obvious relation $2\alpha+\beta<2\pi$. \ZZZ Some \EEE more simple trigonometry (see \cite{MMPS-new, FMPS}) shows that 
\begin{equation}\label{alphars}
  \sin\alpha=\sqrt{1-(\sigma/\lambda_2)^2}=2 (\rho/\lambda_2)\sin\left(\frac{\pi}{2\ell}\right)
  \end{equation}
  and that
  \begin{equation}\label{betaz}
\beta=\beta(\alpha,\gamma_\ell):=2\arcsin\left(\sin\alpha\sin\frac{\gamma_\ell}{2}\right),
\end{equation}
 where $\gamma_\ell$ is the internal angle of a regular $2\ell$-gon,  \PPP i.e., \EEE
 \begin{equation}\label{gamma}
\gamma_\ell:=\pi\left(1-\frac{1}{\ell}\right).
\end{equation}  

 \begin{figure}[htp]
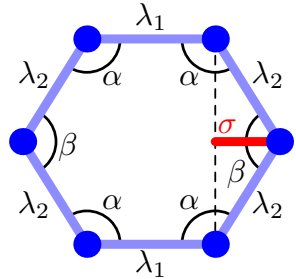

\begin{center}
\pgfdeclareimage[width=0.35\textwidth]{cell}{cell2.pdf}
    \pgfuseimage{cell}
\caption{\PPP The bond lengths and the angles for the hexagon of a configuration in \PPP $\olive\mathscr{F}\EEE$ are represented. A segment representing $\sigma$ (lying on the horizontal diagonal) is drawn in red. The hexagon is in fact not planar but kinked along the horizontal diagonal.  \EEE}
\label{extra}
\end{center}
\end{figure}

%
%

\subsection{Configurational energy}\label{energy_def}
{
We shall define a {\it configurational energy} for  collections of particles of the form 
$\mathcal{C}:=C_n+L\mathbf e_1\Zz$
where $L>0$ plays the role of the macroscopic period of  repetition of the { $n$-sample} $C_n:=\{x_1,\dots,x_n\}$ of $\mathcal C$ and $\mathbf e_1$ is a given direction in $\mathbb R^3$. $C_n$ is a collection of $n$ points $x_i\in\Rz^3$ such that $x_i\cdot \mathbf e_1\in [0,L)$. 
Any such generic configuration  $\mathcal{C}$ is characterized by its $n$-sample $C_n$ and by $L$, so that   we always identify  $\mathcal{C}$ with the couple $(C_n, L)$.  
 }

We now introduce the  energy $E$ for the configuration $\mathcal C$ 
   and \PPP we detail the hypotheses 
 on $E$ that we assume throughout the paper. \EEE We aim here at minimal
assumptions in order \PPP to include \EEE in the analysis most of the
choices that have been used for \PPP computational chemistry \EEE codes
\cite{Brooks83,Clark89,Gunsteren87,Mayo90,Weiner81} and that have been proposed as empirical potentials for the description of covalent systems \cite{Brenner90,Stillinger-Weber85,Tersoff, T2, T3}. In particular, we will not give explicit forms for the \olive energy,  but \EEE only consider potential terms that satisfy some qualitative properties. This will lead to quite general results that are grounded only on natural symmetry, convexity, and monotonicity assumptions. 

The energy $E$ features \olive {\it   two-body} and {\it three-body}  interactions among particles, \EEE that are  respectively  modeled by \olive a potential $v_2$ favoring  bonds between nearest neighbors of unitary length and a potential $v_3$ favoring angles of $120^\circ$ between such  bonds, \EEE see \eqref{E}. 
This is consistent \EEE with the modeling of $sp^2$ covalent
atomic bonding in carbon  favoring  a specific
interatomic distance and angles of $120^\circ$. 

We introduce the   \olive two-body potential \EEE $v_2: \PPP (0,\infty) \EEE\to[-1,\infty)$, which  is smooth and attains its minimum value \RRR only \EEE at $1$ with $v_2(1) = -1$ and  $v''_2(1)>0$. Moreover, we ask $v_2$ to be {\it short-ranged}, i.e., to vanish after some reference value right to $1$, say    $v_2(r)=0$
for $r\ge 1.1$. 
In this way,
 $x,y\in\mathcal{C}$ will contribute to the two-body term of the energy only if
 their distance is less than $1.1$. 
  Looking at the $n$-sample $C_n$ \ZZZ and taking \EEE periodicity into account, \PPP  this amounts to considering two particles $x_i$ and $x_j$ of the $n$-cell $C_n$ of $\mathcal{C}$ as contributing to the two-body part of the energy only if \EEE $|x_i-x_j|_L<1.1$, where $|\cdot|_L$ is the { distance modulo $L$ in the $\mathbf e_1$ direction} defined by
$$
|x_i-x_j|_L:=\min_{t\in\{-1,0,+1\}}|x_i-x_j+Lt\mathbf e_1|
$$
for every $x_i,x_j\in C_n$. We also let 
$
\mathcal{N}:=\{(i,j)\,:\,\, \textrm{$x_i$, $x_j\in C_n$, $i\neq j$, and $|x_i-x_j|_L<1.1$}\}
$ be the set of first-neighboring atoms.

We introduce the \olive three-body potential \EEE
$v_3: [0,2\pi]\to[0,\infty)$
and we assume \olive that $v_3$  is \EEE smooth and  symmetric around $\pi$, i.e., 
$v_3(\alpha)=v_3(2\pi{-}\alpha)$. Moreover, we suppose that the minimum value $0$ of $v_3$
is attained only at $2\pi/3$ and $4\pi/3$ with $v_3''(2\pi/3)>0$. We let 
$
\mathcal{T}:=\{(i,j,k)\,:\,\, \textrm{$i\neq k$, $(i,j)\in\mathcal{N}$ and $(j,k)\in\mathcal{N}$}\}.
$ 
For all triples $(i,j,k)\in\mathcal{T}$ we denote by  $\alpha_{ijk} \in  \RRR [0,\pi] \EEE $
the {\it bond angle} formed by the vectors $x_i-x_j$ and
$x_k-x_j$. 

\olive The \EEE configurational energy $E$ of  $\mathcal{C}=(C_n, L)$ \olive is defined \EEE by
\begin{equation}\label{E}
E(\mathcal{C})=E(C_n,L):=\frac12\sum_{(i,j)\in \mathcal{N}}v_2(|x_i{-}x_j|_L) + \frac12\sum_{(i,j,k)\in\mathcal{T}}v_3(\alpha_{ijk}),
\end{equation} 
\PPP where the factors $1/2$ are included to avoid double-counting \ZZZ of \EEE the interactions among same atoms. \ZZZ Since \EEE the energy is invariant \ZZZ under \EEE isometries, all statements involving $E$ in the following \ZZZ are \EEE considered to hold up to isometries.



For  \PPP a \EEE fixed integer $\ell>3$, let us consider a configuration $\mathcal{F}$  in the family
$\olive\mathscr{F}\EEE$. \PPP As $\mathcal{F}$ is periodic, it can be identified with the couple \EEE $(F_n, L)$, where $F_n$ is the corresponding $n$-cell ($n=4m\ell$ for some $m\in \mathbb{N}$),
and
\begin{equation}\label{zigzagperiod}
L = L^\mu_m:= m\mu
\end{equation}
can be seen as the length of the $n$-sample (notice that for $m=1$ we get the minimal period of the configuration).
We note that it is natural to consider the regime $m\gg\ell$, which \ZZZ accounts for \EEE the fact \olive that \EEE the nanotube is a thin and long structure.
 In view of \eqref{E} and \eqref{betaz}, the energy can be written as 
\begin{align}\label{basicenergy}
E(\mathcal{F}) = E(F_n, L^\mu_m) = \frac{n}{2} \big( v_2(\lambda_1) + 2v_2(\lambda_2)  \big) + n\big(2v_3(\alpha) + v_3(\beta(\alpha,\gamma_\ell))\big).
\end{align}
In the particular case $\lambda_1=\lambda_2=1$, we can reduce to a function $\mathcal E$ of the angle $\alpha$ only, i.e.,
\begin{equation*}\label{=1}
E(\mathcal F)=\mathcal E (\alpha):=-\frac32 n+n(2 v_3(\alpha)+v_3\left(\beta(\alpha, \gamma_\ell))\right).
\end{equation*}

{\BBB 
%

\section{Main local minimality results}\label{sec: mainresults}

\subsection{Minimization in the family $\olive\mathscr{F}_1\EEE $}
Let us take the minimal period $\mu$ for a configuration in $\olive\mathscr F\EEE$ in some reference interval around the value $3$, 
say $\mu\in (2.6,3.1)$. We notice that $3$ is the minimal period in case of unit bonds $\lambda_1=\lambda_2=1$ and in case of a planar structure (which is ideally realized by $\ell=+\infty$). 
The minimization problem
\[
\min\{E(\mathcal F): \mathcal F =\ZZZ \mathcal F_{\lambda_1,\lambda_2,\mu},\; \EEE \lambda_1\in(0.9,1.1),\;\lambda_2\in(0.9,1.1)\}
\]
 is solved by $\lambda_1=\lambda_2=1$.
\RRR This is clear \olive by \EEE \eqref{basicenergy}, by changing   $\lambda_1,\lambda_2$ to $1$ and by leaving $\alpha$ unchanged.    
\ZZZ We  \EEE stress that in case $\lambda_1=\lambda_2=1$, thanks  to  \eqref{eq: basic constraints} and \eqref{alphars}, we have $\mu = 2 (1-\cos\alpha)$, so that the one-parameter family $\olive\mathscr{F}_1\EEE $ can be reparametrized in terms of $\alpha$.
Therefore, we have in principle the two equivalent minimization problems 
\[
\min\{E(\mathcal F): \mathcal F\in \olive\mathscr{F}_1\EEE ,\;\mu\in(2.6,3.1)\},\qquad\min \{\mathcal E(\alpha): \alpha\in(\arccos(-0.3),\arccos(-0.55))\}.
\]
The latter {\BBB  minimization problem \ZZZ in one variable \EEE for the map $\alpha\mapsto \mathcal E(\alpha)$ has} been investigated in \cite[Theorem 4.3]{MMPS}. 
It may require \ZZZ to work \EEE on a smaller interval of values of $\alpha$ around $2\pi/3$, depending on the specific choice of $v_3$. The result in  \cite[Theorem 4.3]{MMPS} is the following\olive:\EEE

\begin{proposition}\label{eq: old main result}
There exist an open interval $A$ and $\ell_0 \in \Nz$ only depending on  $v_3$  such \PPP  that \EEE the following holds for all $\ell \ge \ell_0$:  There is a unique angle $\alpha^{\rm us}_\ell \in A$ such that $\mathcal{E}(\alpha^{\rm us}_\ell)$ minimizes $\mathcal E$ on ~$A$. 
\end{proposition}  
The 
 minimal period of the nanotube corresponding to $\alpha_\ell^{\rm us}$  is given by 
$\mu^{\rm
  us}_\ell := 2 - 2\cos\alpha^{\rm us}_\ell$,
  so that 
 $\mathcal{F}_{1,1,\mu^{\rm
    us}_\ell}$ is the optimal configuration in $\olive\mathscr{F}_1\EEE $
   for $\mu$ varying in a suitable interval around $3$ (depending on $v_3$). 
      {\BBB Nanotubes with $\mu=\mu_\ell^{\rm us}$ will be referred to as  {\it unstretched}, namely {\it (us) nanotubes}, as they  feature unit bond lengths.
      In contrast, later we will consider optimal  strechted nanotubes with different bond lengths.
      
    }



The above result is proven easily by means of the convexity and monotonicity assumptions \ZZZ on  \EEE $v_3$ around $2\pi/3$ and \ZZZ the properties \EEE of the function $\beta(\alpha,\gamma_\ell)$.
We mention that the result in \cite{MMPS} further shows that neither the \PPP {\it polyhedral} \cite{Cox-Hill07, Cox-Hill08} \EEE nor the
classical {\it rolled-up} \cite{Dresselhaus-et-al95} configuration is a local minimizer of the energy $E$. 
The \olive polyhedral and rolled-up configurations are \EEE classical geometric definitions of carbon nanotubes, and in the zigzag case they \olive correspond, \EEE see \cite{MMPS}, by particular choices of $\alpha$, \ZZZ namely   \olive
$$\alpha
= 2\pi/3 \qquad\mbox{and} \qquad \alpha=\beta(\alpha,\gamma_\ell)=\arccos\left(\frac{1-2\sin^2(\gamma_\ell/2)}{2\sin^2(\gamma_\ell/2)}\right), $$ 
respectively. \EEE

We wish to prove that $\mathcal F_{1,1,\mu^{\rm us}_\ell}$
is in fact a local minimizer with respect to any small perturbation. 
More precisely, if $x_1^{\rm us},\ldots x_n^{\rm us}$ are the points in the $n$-sample of $\mathcal F_{1,1,\mu^{\rm us}_\ell}$, we shall prove local minimality with respect to any configuration $\tilde {\mathcal F}$ that belongs to
\begin{align}\label{etaold}
\begin{split}
\mathscr{P}_\eta = \lbrace \tilde{\mathcal{F}} = (F_n,L)| \ F_n := \lbrace x_1,\ldots,x_n \rbrace \ \text{ with } |x_i - x_i^{\rm us}| \le \eta,\; |L-m\mu_\ell^{\rm us}|<\eta \rbrace
\end{split}
\end{align}
for $\eta$  small enough. \ZZZ Here, \EEE $\eta$ is always assumed to be so \olive small that \EEE the first neighboring relations of the perturbed configurations are unchanged, and no new couples of points reach a distance (modulo $L$) which is smaller than $1.1$, \olive that would indeed result in further contributions to the energy. \EEE

A first crucial role towards local minimality  is played by \olive adopting a \EEE energy summation strategy. We may compare two different approaches. The first one is the summation of the energy \olive point-by-point \EEE , which leads to a very clear and simple, although incomplete proof.

  \subsection{\olive Point-by-point \EEE  summation of the energy}\label{pointsection}
 
 Let us denote by $\mathcal F^*=\mathcal F_{1,1,\mu^{\rm us}_\ell}$ the unique optimal nanotube in the family \ZZZ $\olive\mathscr{F}_1\EEE $ \EEE that is provided by Proposition \ref{eq: old main result}, when $\mu$ varies in a suitable neighborhood of the value $3$. We stress again that such neighborhood is not explicit here since it depends on the specific choice of the three-body potential $v_3$ which is left fairly general in our discussion. Indeed, it depends on the amplitude of the interval of strict convexity of $v_3$ around $2\pi/3$. We refer to  \cite{MMPS-new} for a more precise quantification.
 When considering a perturbation $\tilde {\mathcal F}\in\mathscr P_\eta$, see \eqref{etaold}, at each point $x_i$, $i=1,\ldots, n$, of the ($n$-sample of the) perturbed configuration we have
 \begin{itemize}
\item three bond lengths $\bar \lambda_{1}^i$, $\bar\lambda_{2}^i$, $\bar \lambda_3^i$, each in a neighborhood of $1$;
\item
 two bond angles $\alpha_1^i$, $\alpha_2^i$: these are the perturbations of the  angle  $\alpha_\ell^{\rm us}$ of $\mathcal F^*$;
 \item
 the {\it nonplanarity angle} $\gamma^i$,  \MMM i.e., the angle between the two planes that contain respectively $\alpha_1^i$ and $\!\alpha^i_2$ \EEE  ($\gamma^i=\pi$ in the limit  planar case obtained for large $\ell$, and  $\gamma^i=\gamma_\ell$ if $\tilde{\mathcal F}\in\mathscr{F}\EEE$);
 \item a third bond angle $\beta^i$ which can be written as a function of $\alpha_1^i,\alpha_2^i,\gamma^i$ by a simple computation, i.e., $\beta^i=\tilde\beta(\alpha_1^i,\alpha_2^i,\gamma^i)$, where the function $\tilde\beta:[\pi/2,3\pi/4]\times[\pi/2,3\pi/4]\times[3\pi/4,\pi]\to\mathbb R$ is defined \ZZZ by \EEE
 $
\tilde\beta\EEE(\alpha_1,\alpha_2,\EEE \gamma \EEE):=2\arcsin\left(\sqrt{(1-\sin\alpha_1\sin\alpha_2\cos\EEE
    \gamma \EEE-\cos\alpha_1\cos\alpha_2)/2}\right).$
 \end{itemize}
 The angle energy contribution of  each point $x_i$ of the configuration $\tilde{\mathcal F}$ is therefore
 \begin{equation*}
 \EEE \tilde E
\EEE(\alpha_1^i,\alpha_2^i,\EEE \gamma^i \EEE):=v_3(\alpha_1^i)+v_3(\alpha_2^i)+v_3(
\EEE  \tilde\beta \EEE(\alpha_1^i,\alpha_2^i,\EEE \gamma^i \EEE)),
\end{equation*}
so that the \olive point-by-point \EEE  energy summation reads
\begin{equation}\label{summation1}
E(\tilde{\mathcal{F}})= \frac12\sum_{i=1}^n \left(v_2(\bar\lambda_1^i)+v_2(\bar \lambda_2^i)+v_2(\bar\lambda_3^i)\right)+\sum_{i=1}^n \tilde E
\EEE(\alpha_1^i,\alpha_2^i,\EEE \gamma^i \EEE).
\end{equation}
\ZZZ We define \EEE $\sigma_0(\gamma):=\pi-\arcsin\left(\frac{\sqrt{3}}{2\sin(\gamma/2)}\right)<\frac23 \pi$ and collect the basic properties that arise from the study of the function $\tilde E$. For the proof we refer to \cite[Section 4]{MMPS-new}. 
\begin{proposition}\label{etilde}

 There exists $\gamma_0\in\mathbb N$ (depending only on $v_3$) such that for any $\gamma\in(\gamma_{0},\pi)$, the mapping 
$(\alpha_1,\alpha_2)\mapsto\tilde E(\alpha_1,\alpha_2,\gamma)=\tilde E(\alpha_2,\alpha_1,\gamma)$ is continuous and  strictly convex on $[\sigma_0(\gamma),2\pi/3]^2$.  Denoting its unique minimizer by  $(\alpha_\gamma,\alpha_\gamma)$, 
 there holds $\alpha_\gamma\in(\sigma_0(\gamma),2\pi/3)$.
Moreover, the map 
$
\gamma\mapsto \tilde E(\alpha_\gamma,\alpha_\gamma,\gamma)
$
 is strictly decreasing and strictly convex on $(\EEE \gamma_0,\pi)$.
\end{proposition}







%
%
%


%
%
%
%
%


Now let $\ell$ be large enough so that $\gamma_\ell\in(\gamma_0,\pi)$.
If we assume that the mean value $\gamma_{\rm mean}$ of the angles $\gamma^i$ is smaller than $\gamma_\ell=\pi-\pi/\ell$, 
starting from \eqref{summation1}, since $v_2\ge -1$ and using
 the convexity and monotonicity properties of $\tilde E$ from  Proposition \ref{etilde}, letting $\alpha^i=(\alpha_1^i+\alpha_2^i)/2$ we obtain

\vspace{-0.2cm}
\begin{equation}\label{oldproof}\begin{aligned}
E(\ZZZ \tilde{\mathcal{F}} \EEE )&
\ge -\frac32n +\sum_{i=1}^n \tilde E(\alpha_1^i,\alpha_2^i,\gamma^i)
\\&\ge
  -\frac32n +\sum_{i=1}^n \tilde E(\alpha_{\gamma^i},\alpha_{\gamma^i},\gamma^i)
\ge -\frac32 n+
n\tilde{E}(\alpha_{\gamma_\ell},\alpha_{\gamma_\ell},\gamma_\ell)= E(\mathcal{F}^*).
\end{aligned}\end{equation}
On the other hand, if $\gamma_{\rm mean}-\gamma_\ell>0$, there is an {\it angle defect} that does not allow to conclude: the last inequality  fails in general. In order to get rid of the angle defect, we shall resort to a different energy summation.

%
%

\subsection{Axial displacement}  Let us now move forward
to the case of {\it stretched} nanotubes, which we define, by now in the family \olive $\mathscr{F}$, \EEE to be the nanotubes with
   $\mu \ZZZ \ge \EEE \mu^{\rm us}_\ell$. 
For   fixed $\mu \in (2.6,3.1)$, $\mu\neq \mu_\ell^{\rm us}$, we consider the minimization problem 
\begin{align}\label{min2} 
E_{\rm min}(\mu)  = \min\big\{ E(\mathcal{F}_{\lambda_1,\lambda_2,\mu})| \ \mathcal{F}_{\lambda_1,\lambda_2,\mu} \in \olive\mathscr{F}_1\EEE , \ \lambda_1,\lambda_2 \in   (0.9,1.1) \big\}.
\end{align}
We obtain the following existence result which is borrowed from \cite{FMPS} and \olive that \EEE generalizes Proposition \ref{eq: old main result}.
\begin{proposition}[Existence and uniqueness of minimizer: General case]\label{th: main1}
There  exist $\ell_0 \in \Nz$ and, for each $\ell \ge \ell_0$, an open interval $M^\ell$ only depending on $v_2$, $v_3$, and $\ell$, with $\mu^{\rm us}_\ell \in M^\ell$, such that for all $\mu \in M^\ell$ there is a unique pair of bondlengths $(\lambda^\mu_1,\lambda^\mu_2)$ such that $\mathcal{F}_{\lambda^\mu_1,\lambda^\mu_2,\mu}$ is a solution of the problem \eqref{min2}.
\end{proposition}
In the following the minimizer is sometimes denoted for simplicity by $\mathcal{F}_\mu^*$. Note
that we have  $\mathcal{F}_{\mu^{\rm us}_\ell}^* =
\mathcal{F}_{1,1,\mu^{\rm us}_\ell}$  by Proposition \ref{eq: old main
  result}. 
 Our aim is to investigate the local minimality of $\mathcal{F}_\mu^*$. To this end, we consider again \emph{general} small perturbations $\tilde{\mathcal{F}}$ of $\mathcal{F}_\mu^*$ with the same bond graph, this time with prescribed stretching. Again, \EEE each atom keeps three and only three bonds, and we can identify the three neighboring atoms of the perturbed configurations with the ones for the configuration $\mathcal{F}_\mu^*$. By $F^\mu_n = \lbrace x^\mu_1,\ldots, x^\mu_n \rbrace$ we  denote the $n$-cell of $\mathcal{F}_\mu^*$ so that $\mathcal{F}_\mu^* = (F_n^\mu, L^\mu_m)$ with $L^\mu_m$ as defined in \eqref{zigzagperiod} for $m \in \Nz$ with $n = 4m\ell$. We define \emph{small perturbations} $\mathscr{P}_\eta(\mu)$ of $\mathcal{F}_\mu^*$ by 
\begin{align}\label{eq: bc}
\begin{split}
\mathscr{P}_\eta(\mu) = \lbrace \tilde{\mathcal{F}} = (F_n,L^\mu_m)| \ F_n := \lbrace x_1,\ldots,x_n \rbrace \ \text{ with } |x_i - x_i^\mu| \le \eta \rbrace,
\end{split}
\end{align}
where $L^\mu_m$ is defined by \eqref{zigzagperiod}.
  In particular, by choosing  $\eta$ small enough (depending on $v_2$, $v_3$ and $\ell$), the first neighboring relations among atoms do not change and still only first neighbors contribute to the \olive energy. 
 Moreover, \EEE we recall
$E(\tilde{\mathcal{F}}) = E(F_n, L^\mu_m)$. 

Starting from $\mathcal F_{1,1,\mu^{\rm us}_\ell}$, we are imposing
tensile   stress on \ZZZ a \EEE nanotube  by simply modifying $\mu$.   For $\mu>\mu_\ell^{\rm us}$, we are stretching the structure with an axial traction. \ZZZ Note that we consider \EEE perturbations which \ZZZ in general \EEE preserve only the parameter $L=L_m^\mu$.
This action on the structure is very general and includes for instance imposed Dirichlet boundary
conditions, where only the first coordinate of the boundary atoms is prescribed.  
Therefore, we are in the setting of a uniaxial strain experiment where we wish to check the elastic behavior of the structure by validating the Cauchy-Born rule.
Notice that prescribing a value $L^\mu_m$ (with the traction constraint $L^\mu_m>m\mu_\ell^{\rm us}$) can be seen as prescribing the macroscopic length of the $n$-points sample $F_n$. 
Another example of boundary displacement that is compatible with our definition of admissible perturbations of $\mathcal F^*_\mu$
 is the following: \olive All  \EEE atoms that lie in the leftmost section of the $n$-sample of $\mathcal F^*_\mu$ (say the $\ell$ atoms $x$ such that $x\cdot\mathbf e_1$=0) keep lying on the same plane. In this case, the cell length $L^\mu_m$ is the distance between such plane and the parallel one at distance $L^\mu_m$.
\EEE

We now state  the main
result.

\begin{theorem}[Local stability of minimizers]\label{th: main3}
There  exist $\ell_0 \in \Nz$ and for each $\ell \ge \ell_0$ some $\mu^{\rm crit}_\ell > \mu^{\rm us}_\ell$ and $\eta_\ell >0$ only depending on $v_2$, $v_3$, and $\ell$  such that  for all $\ell \ge \ell_0$ and  for all  $\mu \in [\mu_\ell^{\rm us},\mu_\ell^{\rm crit}]$   we have
$E(\tilde{\mathcal{F}})>E(\mathcal{F}_\mu^*) $
for any nontrivial perturbation $\tilde{\mathcal{F}} \in \mathscr{P}_{\eta_\ell}(\mu)$  of the configuration $\mathcal{F}_\mu^*$.
\end{theorem}

 Under prescribed and small stretchings \PPP (i.e., the value of $L_m^\mu$ is prescribed), \EEE
 \PPP we are proving that  there exists a periodic strict-local minimizer  $\mathcal{F}_\mu^*$ that  belongs \EEE  to the
family $\olive\mathscr{F}_1\EEE $. 
This can be seen as a validation of the
Cauchy-Born rule  in  this specific setting. \EEE
The result applies in particular to the case $\mu=\mu^{\rm us}_\ell$, which implies 
the validity of the local minimality result that we have previously discussed in Section \ref{pointsection} as a particular case. \ZZZ More specifically, we recall that the proof provided in Section \ref{pointsection} was incomplete, but is now completed by means of the main result of our work. \EEE We also mention that fracture has to be expected for $\mu$ above the critical value $\mu_\ell^{\rm crit}$, i.e., if traction  is too strong, see \cite{FMPS}.

\section{Proof of the main result}\label{sec: main proof}
This section provides an overview of the main lines of the proof of Theorem \ref{th: main3}, following the arguments in \cite{FMPS}. The first point is a different energy summation with respect to the one we adopted in Section \ref{pointsection}.

\subsection{\olive Cell-by-cell \EEE  summation of the energy}
For the configuration $\mathcal F^*_\mu$ from Proposition \ref{th: main1}, the $n$-points sample $F_n$ is made by $4m$ sections that are orthogonal to the \ZZZ axis of the cylinder \EEE (each with $\ell$ atoms). We consider a subdivision of $F_n$ in {\it basic cells} \olive (or, in the following, simply  cells). \EEE

 A basic cell is made of eight points which correspond to a hexagon plus its two first neighboring atoms along the parallel to the axis, see Figure  \ref{cell}.  We denote by $\boldsymbol{x}=(x_1,\ldots, x_8)$ the generic cell, with the ordering from Figure \ref{cell}, where $x_2-x_1$ is parallel to the axis vector $\mathbf e_1$. Notice that the whole $n$-points sample contains exactly $2m\ell$ cells, that are spanned by letting the indices in the definition of $\mathscr F(\lambda_1,\lambda_2,\mu)$ vary as $i=1,\ldots, \ell$, $j=1,\ldots m$, $k=0,1$ (while the index $l\in\{0,1\}$ varies within the same cell).

 We note that about the ends of the $n$-points sample, we always consider the modulo $L^\mu_m$ repetition of points in the $\mathbf e_1$ direction, \ZZZ i.e., we leave \EEE  the interior cells and the boundary cells undistinguished. This convention is made in all the next definitions.
\EEE
 
We define \olive the {\it cell center} and the {\it dual center} by
\begin{align}\label{eq: dual center}
 z:=\frac{x_1+x_2}{2} \qquad\mbox{and } \qquad z^{\rm dual}=\frac{x_2+x_8}2,
 \end{align}
 respectively, \EEE see Figure \ref{cell}.
A small perturbation $\tilde {\mathcal F}$ in $\mathscr P_{\eta_\ell}(\mu)$ of $\mathcal F^*_\mu$ does not change first-neighbor relations. Therefore, we can define basic cells also in $\tilde {\mathcal F}$ by the identification of the points of $\tilde {\mathcal F}$ with the ones of $\mathcal F^*_\mu$. Given the general configuration \olive $\tilde {\mathcal F}$ and any of its cells $\boldsymbol x$, the 
 \begin{figure}[htp]
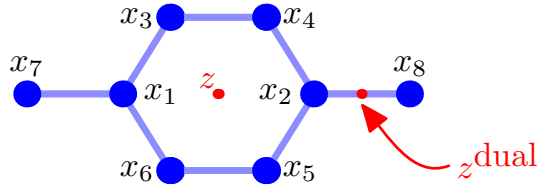

\begin{center}
\pgfdeclareimage[width=0.6\textwidth]{cell}{cell_new.pdf} 
    \pgfuseimage{cell} 
\caption{ Notation for the points and the centers in the basic cell.}
\label{cell}
\end{center}
\end{figure}
bond lengths $b_i$, $i=1,\ldots, 8$, and the \EEE bond angles $\varphi_i$, $i=1,\ldots 10$, are defined and labelled as shown in Figure \ref{cellangles}.
 \begin{figure}[htp]
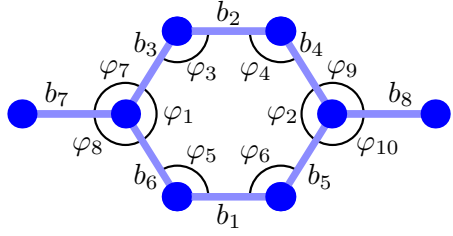

\begin{center}
\pgfdeclareimage[width=0.4\textwidth]{cellangles}{cellangles.pdf}
    \pgfuseimage{cellangles}
\caption{ Notation for the bond lengths and angles in the basic cell.}
\label{cellangles}
\end{center}
\end{figure}

The \emph{cell energy} is the energy contribution of a single  cell $\boldsymbol x$ and it is defined 
by \begin{equation}
\begin{aligned} 
E_{{\rm cell}}(\boldsymbol x) & = \frac{1}{4} \big(v_2(b_1) + v_2(b_2) \big) + \frac{1}{2}\sum^6_{\PPP h\EEE=3} v_2(b_{\PPP h\EEE})    + \frac{1}{4} \big( v_2(b_7) + v_2(b_8) \big)
\\ 
&  \qquad + v_3(\varphi_1) + v_3(\varphi_2) + \frac{1}{2}\sum^6_{\PPP h\EEE=3}
v_3(\varphi_{\PPP h\EEE})  +\frac{1}{2} \sum^{10}_{\PPP h\EEE=7}
v_3(\varphi_{\PPP h\EEE}). \label{eq: cell} 
\end{aligned}\end{equation}
The coefficients in the definition of $E_{\rm cell}$ are taken by considering that each bond that is not (approximately) parallel to $\mathbf e_1$ is contained exactly in two cells and each of the other bonds  is contained in four cells,  and similarly for bond angles. \ZZZ The \EEE computation of the energy of the cell $\boldsymbol x$ requires the knowledge of the mapping $T:\mathbb R^{24}\to\mathbb R^{18}$ that associates to $\boldsymbol x=(x_1,\ldots, x_8)\in\mathbb R^{24}$ the corresponding bond lengths $b_i$ and bond angles $\varphi_i$ that appear in the  above right hand side.


When considering the whole $n$-points sample, we index \ZZZ its generic cells \EEE as $\boldsymbol x_{i,j,k}$, $i=1,\ldots \ell$, $j=1,\ldots,m$, $k=0,1$. Therefore, the \olive cell-by-cell \EEE  summation formula  for the energy reads
 
\begin{align}\label{eq: sumenergy}
E(\tilde{\mathcal{F}}) = \sum_{i=1}^\ell \sum_{j=1}^m \sum_{k=0,1} E_{{\rm cell}}(\boldsymbol x_{i,j,k}). 
\end{align}

 \PPP
As we have seen in Section \ref{pointsection} (at least for $\mu=\mu_\ell^{\rm us}$), the angle defect is the main obstacle to the conclusion of the proof of Theorem \ref{th: main3}.
Since the strategy is now to sum the energy cell by cell, \olive we  introduce \EEE the notion of \emph{nonplanarity angle} of a cell, which generalizes the nonplanarity angle from Section \ref{pointsection}.

For a cell $\boldsymbol{x}=(x_1,\ldots,x_8)$ with center $z$ and dual center $z^{\rm dual}$, we denote by $\theta_l(z)$ the angle between the plane through $x_1,x_3,x_4$ and the one through $x_1,x_5,x_6$ \olive (following the labelling of Figure \ref{cell}). \EEE Similarly, we define $\theta_r(z)$ as the angle between the plane through $x_2,x_3,x_4$ and the one through $x_2,x_5,x_6$. Moreover, $\theta_l(z^{\rm dual})$ is the angle between the plane through ${x_2,x_4,x_8}$ and the one through $x_2,x_5,x_8$. Finally, $\theta_r(z^{\rm dual})$ is the angle between the plane through  $x_2,x_8,x_*$ and the one through $x_2,x_8,x_{**}$ (here, $x_*$ and $x_{**}$ are the other two first neighboring atoms of $x_8$: they belong to another cell). For the sake of definiteness, all these angles are in $[0,\pi)$.

The {\it nonplanarity angle}  of a cell $\boldsymbol x$ with center $z$ and dual center $z^{\rm dual}$ is defined \begin{equation}\label{nonplanarity2}\bar\theta(\boldsymbol{x}):=\frac14\left(\theta_l(z)+\theta_r(z)+\theta_l(z^{\rm dual})+\theta_r(z^{\rm dual})\right).\end{equation}

\noindent
\subsection{Symmetrization of cells} 
Let us now consider the basic cell of the configuration $\mathcal F^*_\mu$ that is defined by Proposition \ref{th: main1} (of course, all the basic cells of such configuration are equal).
We denote it by $\boldsymbol x^{\rm kink}$, where the terminology is hinting to the fact that it is kinked along the \ZZZ diagonal of the  hexagon \EEE whose direction is $\mathbf e_1$. Indeed, the eight points of $\boldsymbol x^{\rm kink}$ are contained in two planes.  
We can give the precise position of the eight points by fixing a reference orthogonal system $\mathbf e_1,\mathbf e_2,\mathbf e_3$    with  the cell center coinciding with the origin. We let $\mathbf e_1$ be axis direction as usual, we let 
 $\mathbf e_2$ be  the direction of $x_3-x_6$,  and we let  $\mathbf e_3=\mathbf e_1\wedge \mathbf e_2$.
 The exact positions of the points in $\boldsymbol x^{\rm kink}$ are therefore
  \begin{equation*}\label{kink} 
 \begin{aligned}
&x_1^{\rm kink} = (-1/2-\sigma^{\rm us},0,0),\quad   x^{\rm kink}_2 = (1/2 + \sigma^{\rm us},0,0), \\
&  x_3^{\rm kink} =(-1/2,\sin\alpha^{\rm us}_\ell \sin(\gamma_\ell/2),\sin\alpha^{\rm us}_\ell \cos(\gamma_\ell/2)),\quad
  x_4^{\rm kink} =
  x_3^{\rm kink}+(1,0,0),  \\
&  x_5^{\rm kink} =(1/2,-\sin\alpha^{\rm us}_\ell \sin(\gamma_\ell/2),\sin\alpha^{\rm us}_\ell \cos(\gamma_\ell/2)), \quad
 x_6^{\rm kink} 
 =x_5^{\rm kink}-(1,0,0), 
 \\& x_7^{\rm kink} = (-3/2 - \sigma^{\rm us},0,0),\quad  x_8^{\rm kink} = (3/2 + \sigma^{\rm us},0,0),
\end{aligned}
 \end{equation*}
\ZZZ where $\sigma^{\rm us}$ corresponds to the parameter in \eqref{eq: basic constraints}, $\gamma_\ell$ is defined in \eqref{gamma}, and $\alpha^{\rm us}_\ell$ as given in Proposition \ref{eq: old main result}. \EEE   For a generic perturbation $\boldsymbol x$ of $\boldsymbol x^{\rm kink}$ (resulting from a perturbation $\widetilde{\mathcal F}\in\mathscr P_\eta(\ZZZ \mu \EEE )$), we can consider the same reference system,
   and  by adding a rigid motion we may reduce \ZZZ without restriction to the following situation: \EEE  the second and third components of  $(x_1+x_7)/2$, $(x_2+x_8)/2$ (which are the dual centers of $\boldsymbol x$ and of an adjacent cell) are equal to zero \EEE and the points $x_4$, $x_5$ lie in a plane that is parallel to the one generated by  \ZZZ $\mathbf  e_1$ and $\mathbf e_2$. \EEE  
For  $y = (y^1, y^2,y^3) \in \Rz^3$ we \olive consider \EEE $r_1 (y) := (-y^1,y^2,y^3)$ and $r_2 (y) := (y^1,-y^2,y^3)$. For the generic cell $\boldsymbol {x}= (x_1,\ldots,x_8)$, we define the reflections
\begin{equation*}\label{reflexion}
\begin{aligned}
S_1(\boldsymbol{x})& = ( r_2(x_1) \, | \,  r_2(x_2)   \, | \, r_2(x_6)  \, | \,  r_2(x_5) \, | \, r_2(x_4) \, | \, r_2(x_3) \, | \, r_2( x_7) \, | \, r_2(  x_8)),\\ 
S_2(\boldsymbol{x}) &= ( r_1( x_2) \, | \, r_1( x_1) \, | \,  r_1(x_4 )\,  | \,  r_1( x_3)  \, | \,r_1( x_6) \, | \,  r_1(x_5) \, | \, r_1( x_8) \, | \, r_1( x_7)).
\end{aligned}
\end{equation*}
\BBB
We  define the {\it reflected perturbations}
\begin{equation*}\label{s1s2}
\boldsymbol{x}_{S_1}:=\boldsymbol{x}_{\rm kink}^\ell+S_1(\boldsymbol{x}-\boldsymbol{x}_{\rm kink}^\ell),\quad
\boldsymbol{x}_{S_2} : = \boldsymbol{x}_{\rm kink}^\ell  + S_2(\boldsymbol{x}- \boldsymbol{x}_{\rm kink}^\ell).
\end{equation*}
 Indeed, since $\boldsymbol{x}$ is  a perturbation of  $\boldsymbol{x}^{\rm kink}$, $\boldsymbol{x}_{S_1}$ \PPP(resp.~$\boldsymbol{x}_{S_2}$) \EEE is the \PPP reflected \EEE perturbation \olive with respect to \EEE the plane generated by ${\mathbf e_1, \mathbf e_3}$ \PPP (resp.~$\mathbf e_2, \mathbf e_3$).  \EEE We note that  $E_{\rm cell}(\boldsymbol{x}_{S_2} ) =  E_{\rm cell}(\boldsymbol{x}_{S_1} ) =  E_{\rm cell}(\boldsymbol{x})$ is a consequence of the symmetry of the configurations.
 \ZZZ We now define the \emph{symmetrization of the cell} \EEE in two steps. \olive We  first perform a symmetrization with respect to the plane generated by ${\mathbf e_1, \mathbf e_3}$, then  a symmetrization with respect to the plane generated by $\mathbf e_2, \mathbf e_3$. \EEE Of course, $\boldsymbol{x}^{\rm kink}$ is itself symmetric \olive with respect to \EEE these two planes. 
  We define \PPP the \emph{symmetrized} perturbations \EEE by 
\begin{equation*}\begin{aligned}\label{reflection2}
    \boldsymbol{x}' &:=  \boldsymbol{x}^{\rm kink} + \frac{1}{2} \Big((\boldsymbol{x}  - \boldsymbol{x}^{\rm kink})   + S_1(\boldsymbol{x}- \boldsymbol{x}^{\rm kink}) \Big),\\
    \mathcal{S}(\boldsymbol{x}): &= \boldsymbol{x}^{\rm kink} +
    \frac{1}{2} \Big((\boldsymbol{x}' - \boldsymbol{x}^{\rm
      kink}) + S_2(\boldsymbol{x}'- \boldsymbol{x}^{\rm
      kink}) \Big).
  \end{aligned}
\end{equation*}
We finally define the \emph{symmetry defect} as the following quadratic deviation
\begin{align}\label{delta}
\Delta(\boldsymbol{x}) := |\boldsymbol{x} - \boldsymbol{x}'|^2 + |\boldsymbol{x}' - \mathcal{S}(\boldsymbol{x})|^2.
\end{align}
We notice that the distance among the points $\frac12(x_1+x_7)$ and $\frac12(x_2+x_8)$ from the cell  $\boldsymbol{x}$ does not change in passing to the symmetrized configuration $\mathcal S(\boldsymbol{x})$, i.e.,
\begin{equation*}\label{neu}
\left|{(\mathcal S(\boldsymbol{x})_1+\mathcal S(\boldsymbol{x})_7)}-{(\mathcal S(\boldsymbol{x})_2+\mathcal S(\boldsymbol{x})_8)}\right|
=\left|{(x_1+x_7)}-{(x_2+x_8)}\right|.
\end{equation*}


%
%
%

The cell energy  from \eqref{eq: cell} is expressed  as a function of the $18$ variables $b_1,\ldots b_8$, $\varphi_1,\ldots,\varphi_{10}$. After taking symmetrization as above,  the number of variables is reduced. In fact, let us consider those cells $\boldsymbol{x}$ that are symmetrized by the above procedure. \ZZZ Then, possibly up to an additional rigid motion, \EEE there are lengths $\lambda,\tilde\mu$ and angles $\alpha_1,\alpha_2,\beta,\gamma_1,\gamma_2$ such that
\[
\begin{aligned}
&\varphi_1 = \varphi_2 = \beta,\quad \varphi_3 = \varphi_4 = \varphi_5 = \varphi_6 = \alpha_1, \quad \varphi_7 = \varphi_8 = \varphi_9=  \varphi_{10} = \alpha_2, \\
&(x_2+x_8)-(x_1+x_7)=2\tilde\mu\mathbf e_1, \quad b_1=b_2=b_7=b_8=\tilde\mu/2+\lambda\cos\alpha_1,\quad b_3=b_4=b_5=b=6=\lambda
\end{aligned}
\]
and such that ${\theta_l}(\ZZZ z \EEE ) = {\theta_r}(\ZZZ z \EEE   ) = \gamma_1$. Moreover, the angle between the plane through $x_7,x_1,x_3$ and the one through $x_7,x_1,x_6$ is $\gamma_2$ and the angle between the plane through \ZZZ $x_2,x_8,x_4$ \EEE and the one through $x_2,x_8,x_5$ is $\gamma_2$, as well.

%
%
The notation for \olive $\alpha_1,\alpha_2$ and \EEE $\beta$ corresponds indeed to the angle \olive notations \EEE from Section \ref{Fsection}, and  $\beta$ can be expressed by the same simple trigonometric relations, \olive namely \EEE
\begin{align*}
\beta= \beta(\alpha_1,\gamma_1) =  2\arcsin\left(\sin\alpha_1\sin\frac{\gamma_1}{2}\right) = \beta(\alpha_2,\gamma_2) =   2\arcsin\left(\sin\alpha_2\sin\frac{\gamma_2}{2}\right).
\end{align*}

{\BBB We notice that  $\tilde \mu=\mu$ for a basic cell of a nanotube in $\olive\mathscr{F}_1\EEE $.}   Since we are always taking a small perturbation of an optimal configuration  $\mathcal F^*_\mu=\mathcal{F}_{\lambda_1^\mu,\lambda_2^\mu,\mu}$, the values of all the parameters are close to the corresponding values of $\mathcal{F}_{\lambda_1^\mu,\lambda_2^\mu,\mu}$.
%
For a highly symmetric cell as above, the cell energy \eqref{eq: cell} is equal to 
\begin{align}\label{symmetric-cell}
\begin{split}
E^{\rm sym}_{\tilde\mu,\ZZZ \gamma_1,\gamma_2 \EEE }(\lambda,\alpha_1,\alpha_2)& \ZZZ := \EEE 2v_2(\lambda)    + \frac{1}{2} {v}_2  \big( \tilde\mu/2 + \lambda\cos\alpha_1 \big) + \frac{1}{2} {v}_2  \big(\tilde\mu/2 + \lambda\cos\alpha_2 \big)  \\
& \ \ \ \  +  2   v_3(\alpha_1)  +  2 v_3(\alpha_2)   +  v_3(\beta(\alpha_1,\gamma_1)) +  v_3(\beta(\alpha_2,\gamma_2)).
\end{split}
\end{align}
Therefore, the energy of highly symmetric cells as above is a function of six the variables $\alpha_1$, $\alpha_2$, $\gamma_1$, $\gamma_2$, $\lambda$, $\tilde\mu$.
We further introduce the {\it reduced energy}, that is obtained by minimizing in the three variables $\alpha_1,\alpha_2,\lambda$, more precisely
  \begin{align*}
E_{\rm red}(\mu,\gamma_1,\gamma_2) &:= \min\lbrace E_{\mu,\gamma_1,\gamma_2}^{{\rm sym}}(\lambda,\alpha_1,\alpha_2)| \ \lambda \in (0.9,1.1), \ \alpha_1,\alpha_2 \in (\arccos(-0.4),\arccos(-0.55)) \rbrace.
\end{align*}
Here, the optimization interval are chosen for the sake of definiteness only.
Since $E_{\mu,\gamma_1,\gamma_2}^{{\rm sym}}$ is symmetric in  $(\alpha_1,\gamma_1)$ and $(\alpha_2,\gamma_2)$, then $E_{\rm red}$ is symmetric in $\gamma_1$ and  $\gamma_2$, \PPP i.e., \EEE $E_{\rm red}(\mu,\gamma_1,\gamma_2)  = E_{\rm red}(\mu,\gamma_2,\gamma_1)$.  
We stress that the parameter $\mu$ in $E_{\rm red}$ plays the role of the stretching parameter, while $\gamma_1,\gamma_2$ are again nonplanarity measures of the cell. We borrow from \cite{FMPS} a general result concerning the properties of $E_{\rm red}$. It can be thought as a generalization of Proposition \ref{etilde} \ZZZ featuring  convexity and monotonicity properties. \EEE 
It is proven after some lengthy explicit computations for which we refer to \cite[Section 6]{FMPS}.
\begin{proposition}\label{th: mainenergy}
There exists $\ell_0 \in \Nz$ and for each $\ell \ge \ell_0$ there are open intervals $M^\ell$, $G^\ell$ only depending on $v_2$,  $v_3$ and $\ell$ with $\mu^{\rm us}_\ell \in M^\ell$, $\gamma_\ell \in G^\ell$ \PPP 
\EEE such that the following holds:   
\begin{itemize}
\item[1.] \BBB For each   $(\mu, \gamma_1,\gamma_2) \in M^\ell \times G^\ell \times G^\ell$  there exists a unique triple $(\lambda^\mu, \alpha^\mu_1,\alpha^\mu_2)$ solving the minimization problem which defined $E_{\rm red}$. {\BBB   {Moreover, $\alpha_1^\mu=\alpha_2^\mu$ if $\gamma_1=\gamma_2$.}} \EEE
\item[2.]  $E_{\rm red}$ is strictly convex on  $M^\ell \times G^\ell \times G^\ell$, in particular there is a  constant $c_0'>0$ only depending on $v_2$ and $v_3$  such that
$E_{\rm red}(\mu,\gamma_1,\gamma_2) \ge E_{\rm red}(\mu,\bar{\gamma},\bar{\gamma}) + c_0'\ell^{-2} (\gamma_1 - \gamma_2)^2$, with $\bar{\gamma} = (\gamma_1 + \gamma_2)/2$, for all $\mu \in M^\ell$ and $\gamma_1,\gamma_2 \in G^\ell$.
\item[3.]  For each $\mu \in M^\ell$, the mapping $g(\gamma):= E_{\rm red}(\mu,\gamma,\gamma)$ is decreasing on $G^\ell$ with $|g'(\gamma)| \le C\ell^{-3}$ for all $\gamma \in G^\ell$ for some $C>0$ depending only on  $v_3$.
\item[4.]  The mapping $h(\mu):= E_{\rm red}(\mu,\gamma_\ell,\gamma_\ell)$ is strictly convex on $M^\ell$ with $h''(\mu^{\rm us}_\ell)>0$ and strictly increasing on $M^\ell \cap \lbrace \mu \ge \mu^{\rm us}_\ell \rbrace$. 
\end{itemize}
\end{proposition}

We note en passant that once the above properties are established, it is easy to obtain the proof of Proposition \ref{th: main1}.
 For each   $\mu \in M^\ell$ and $\gamma_1 = \gamma_2 = \gamma_\ell$, letting $\lambda_1^\mu = \mu/2 + \lambda^\mu \cos\alpha^\mu_1$ and $\lambda_2^\mu = \lambda^\mu$ with $\lambda^\mu$ and $\alpha^\mu_1$ from Property  {\it 1.} above, \EEE the configuration $\mathcal{F}_{\lambda_1^\mu,\lambda_2^\mu,\mu}$ is the unique minimizer of the problem \eqref{min2}. This can be seen by using \eqref{symmetric-cell} which yields along with the definition of $E_{\rm red}$
\begin{equation}\label{lasty}
E(\mathcal{F}_\mu^*)  = E(\mathcal{F}_{\lambda_1^\mu,\lambda_2^\mu,\mu}) = 2m\ell E_{\rm red}(\mu,\gamma_\ell,\gamma_\ell).
\end{equation}

%

\subsection{Estimates in terms of the symmetry defect}
We denote the unique minimizer from Proposition \ref{th: main1} again by $\mathcal{F}_\mu^*$. For a small perturbation $\tilde{\mathcal F}\in \mathscr{P}_\eta(\mu)$ we consider the usual identification with the $n$-points sample $\tilde{\mathcal F}=(F_n, L^\mu_m)$ and we denote the cells in $F_n$ by the usual three index notation $\boldsymbol{x}_{i,j,k}$, $i=1,\ldots \ell$, $j=1,\ldots, m$, $k=0,1$.
  Based on the  properties of the reduced energy $E_{\rm red}$, \EEE we are able to show that, up to a linear perturbation in terms of the symmetry defect $\Delta$ defined in \eqref{delta}, $E_{\rm red}$ bounds the cell energy $E_{\rm cell}$ from \olive below. \EEE

\begin{theorem}[Energy defect controls symmetry defect]\label{th: Ered}
There exist   $C>0 $ and $\ell_0 \in \Nz$ (only depending on $v_2$ and $v_3$) such that for any $\ell \ge \ell_0$ the following property holds: there exist $\eta_\ell > 0$ and an open interval $M^\ell$ containing $\mu^{\rm us}_\ell $, such that for all $\mu \in M^\ell$ and any perturbation  $\tilde{\mathcal{F}} \in  \mathscr{P}_{\eta_\ell}(\mu)$,    we have \begin{align}\label{mainproof1}
 E_{{\rm cell}}(\boldsymbol{x}_{i,j,k}) \ge  E_{\rm red}\big(| z^{\rm dual}_{i,j,k} - z^{\rm dual}_{i,j-1,k}| , \bar{\theta}(\boldsymbol{x}_{i,j,k}), \bar{\theta}(\boldsymbol{x}_{i,j,k})\big)   + C \ell^{-2}\Delta(\boldsymbol{x}_{i,j,k}).
\end{align}
\end{theorem}

\ZZZ Here, for the definition of the dual centers we refer to \eqref{eq: dual center}. See \eqref{nonplanarity2} for the definition of $ \bar{\theta}(\boldsymbol{x}_{i,j,k})$. \EEE The proof of the above theorem requires some hard computations and it represents a key step towards the proof of the main result. Indeed, thanks to  Theorem \ref{th: Ered}, we are able to quantify the energy gain in passing from a highly non-symmetric configuration to a symmetrized one, in terms of the symmetry defect $\Delta$. Thanks to this gain, it will be possible to get rid of the angle defect that arises when taking the average nonplanarity angle of the rolled-up structure, as seen in Section \ref{pointsection}.  
The proof of Theorem \ref{th: Ered}  is based on a detailed study of the convexity properties of the mapping $T:\mathbb R^{24}\to\mathbb R^{18}$ which associates the position $(x_1,\ldots, x_8)$ of the points of a cell to the associated bond lengths $b_1,\ldots b_8$ and  bond angles $\varphi_1,\ldots \varphi_{10}$.
For the proof we refer to \cite[Section 7]{FMPS}.


When resorting to a \olive cell-by-cell \EEE  energy summation, the nonplanarity angles are defined by \eqref{nonplanarity2}. \ZZZ It \EEE is also possible to prove the following estimate in terms \ZZZ of \EEE $\Delta$. For the proof, we refer to \cite[Section 5]{FMPS}.
\begin{theorem}[Symmetry defect controls angle defect]\label{lemma: sum}
There is a universal constant $c>0$ such that for $\eta>0$ small enough and \RRR for  all \PPP $\tilde{\mathcal{F}}\in \mathscr{P}_\eta(\mu) $ \EEE with cells $\boldsymbol{x}_{i,j,k}$ and $\Delta(\boldsymbol{x}_{i,j,k})\le\eta$ \EEE we have
\begin{align*}
4\sum_{j=1}^m\sum_{i=1}^\ell \sum_{k=0,1}\bar\theta(\boldsymbol{x}_{i,j,k})  \le    4m(2\ell - 2)\pi + c\sum_{j=1}^m\sum_{i=1}^\ell\sum_{k=0,1} \Delta(\boldsymbol{x}_{i,j,k}).
\end{align*}
\end{theorem}
The left hand side above is equal to  $4m(2\ell - 2)\pi$ if $\tilde{\mathcal{F}} \in \olive\mathscr{F}_1\EEE $.

\EEE

\subsection{Proof of the main theorem}
After having collected the main technical auxiliary results in the previous subsections,
we conclude by giving the main line of the proof of Theorem \ref{th: main3}. 

Let  $\ell_0$ be large enough  and let $M^\ell$ be an open interval containing $\mu^{\rm us}_\ell$ such
that Proposition \ref{th: mainenergy} and  Theorem \ref{th: Ered} hold
for all  $\ell\ge \ell_0$ and $\mu \in M^\ell$.  Also let $G^\ell$ be the interval from
Proposition \ref{th: mainenergy} and let $\mu^{\rm crit}_\ell >
\mu^{\rm us}_\ell$ be such that $[\mu^{\rm us}_\ell, \mu^{\rm crit}_\ell]
\subset  M^\ell$. \olive Given  $\ell \ge \ell_0$ and  $\mu \in
[\mu^{\rm us}_\ell, \mu^{\rm crit}_\ell]$, we \EEE consider a
nontrivial perturbation $\tilde{\mathcal{F}} \in
\mathscr{P}_{\eta_\ell}(\mu)$, see \eqref{eq: bc}, with $\eta_\ell$ as in Theorem \ref{th:
  Ered}. 
  We denote again the generic \ZZZ cells \EEE of the $n$-points sample by $\boldsymbol{x}_{i,j,k}$, its center by $z_{i,j,k}$ and its dual center by $z^{\rm dual}_{i,j,k}$. 
The nonplanarity angle $\bar\theta(\boldsymbol{x}_{i,j,k})$ is defined \olive in \EEE \eqref{nonplanarity2}.

\olive Furthermore, we introduce the {\it average cell length} and the {\it average non-planarity angle},  namely \EEE
$$\bar{\mu} = \frac{1}{2m\ell}\sum_{j=1}^m\sum_{i=1}^\ell \sum_{k=0,1} |z^{\rm dual}_{i,j,k} - z^{\rm dual}_{i,j-1,k}|,\qquad\bar{\theta} = \frac{1}{2m\ell}\sum_{j=1}^m\sum_{i=1}^\ell \sum_{k=0,1} \bar{\theta}(\boldsymbol{x}_{i,j,k})
$$
\olive respectively.   By taking the perturbation parameter $\eta_\ell$ small enough,  we \EEE  can assume that 
$\Delta(\boldsymbol{x}_{i,j,k})$ satisfies the assumptions of Theorem \ref{lemma: sum} for all $i,j,k$. \ZZZ We \EEE deduce that 
\begin{equation}\label{av}\bar{\theta} \le \frac{1}{8m\ell}\Big( 4m(2\ell - 2)\pi + c\sum_{j=1}^m\sum_{i=1}^\ell\sum_{k=0,1} \Delta(\boldsymbol{x}_{i,j,k}) \Big) \le   \gamma_\ell + \frac{c}{8m\ell}\sum_{j=1}^m\sum_{i=1}^\ell\sum_{k=0,1} \Delta(\boldsymbol{x}_{i,j,k}).\end{equation}
We note that we can not exclude that  $\bar\theta>\gamma_\ell$, so that  an angle defect appears as seen in \eqref{oldproof}. However, we have obtained a bound in terms of the symmetry defect.

\olive If $\eta_\ell$ is small,  $|z^{\rm dual}_{i,j,k} - z^{\rm dual}_{i,j-1,k}| \in M^\ell$ \EEE and $\bar{\theta}(\boldsymbol{x}_{i,j,k}) \in G^\ell$ for all $i,j,k$. By Theorem \ref{th: Ered}    there is $C>0$  (depending only on $v_2$, $v_3$)  such that for each cell \eqref{mainproof1}  holds.
\ZZZ We now \EEE take advantage of the \olive cell-by-cell \EEE  energy summation \eqref{eq: sumenergy}:
\olive By \EEE \eqref{mainproof1}, by \eqref{eq: sumenergy}, and by using  Property {\it 2.} of Proposition \ref{th: mainenergy}, we find
\begin{align}\label{av2}
E(\tilde{\mathcal{F}}) &= \sum_{i=1}^\ell \sum_{j=1}^m \sum_{k=0,1}  E_{{\rm cell}}(\boldsymbol{x}_{i,j,k})\ge 2m\ell  E_{\rm red}(\bar{\mu}, \bar{\theta},\bar{\theta}) + C \ell^{-2}\sum_{i=1}^\ell \sum_{j=1}^m \sum_{k=0,1} \Delta(\boldsymbol{x}_{i,j,k}).
\end{align}
On the other hand, by \eqref{av}, \eqref{av2}, and by  Property {\it 3.} of Proposition \ref{th: mainenergy} we deduce that for some $C'>0$ only depending { on $v_3$} there holds
\begin{align}\label{mainproof2}
 E(\tilde{\mathcal{F}})   \ge 2m\ell  E_{\rm red}(\bar{\mu}, \gamma_\ell,\gamma_\ell) + \big(C \ell^{-2} - C'\ell^{-3} \big)\sum_{j=1}^m \sum_{i=1}^\ell \sum_{k=0,1} \Delta(\boldsymbol{x}_{i,j,k}).
\end{align}
Notice that the estimates in \olive the \EEE last two lines are analogous  \ZZZ to the  ones  \EEE in \eqref{oldproof} as they make use of minimality, convexity, and monotonicity.
We stress that, in contrast to  \eqref{oldproof}, after passing to the energy of a symmetrized configuration (which is measured by $E_{\rm red}$) we have a gain in terms of the symmetry defect. On the other hand, the angle defect is also estimated in terms of the symmetry defect, and from \eqref{mainproof2} we see that the latter can be overcome if $\ell$ is large enough.

In order to conclude, we need to pass from $\bar\mu$ to $\mu$ with a monotonicity argument.
\olive Notice \EEE that only small perturbations with prescribed value of $L^\mu_m=m\mu$ are allowed, \olive see  \eqref{eq: bc}.  \EEE Therefore, for fixed $i$ and $k$ we have
$$m\mu = L^\mu_m = \Big|\sum_{j=1}^m z^{\rm dual}_{i,j,k} - z^{\rm dual}_{i,j-1,k} \Big| \le  \sum_{j=1}^m |z^{\rm dual}_{i,j,k} - z^{\rm dual}_{i,j-1,k}|. $$
As a consequence, by taking the sum over all $i$ and $k$, we get $\bar{\mu} \ge \mu \ge \mu^{\rm us}_\ell$. Then, from \eqref{mainproof2} and by Property {\it 4.}  of Proposition \ref{th: mainenergy}, we find
\begin{equation}\begin{aligned}
E(\tilde{\mathcal{F}}) &\ge 2m\ell E_{\rm
  red}(\mu,\gamma_\ell,\gamma_\ell) + C''\ell^{-2}\sum_{i=1}^\ell
\sum_{j=1}^m \sum_{k=0,1} \Delta(\boldsymbol{x}_{i,j,k})
\label{mainproof3}
\end{aligned}\end{equation}
for $\ell_0$ sufficiently large and a constant $C''>0$ depending on $v_2,v_3$.  We stress that the assumption  $\mu \ge \mu^{\rm us}_\ell$,  \PPP i.e., \EEE the nanotube is unstretched or under traction but not under compression, is crucial for the application \ZZZ of  the monotonicity property \EEE from Property {\it 4.} of Proposition  \ref{th: mainenergy}.

By  \eqref{lasty} and \eqref{mainproof3} we get
$E(\tilde{\mathcal{F}})\ge E(\mathcal{F}_\mu^*)$.
 Eventually, it is not difficult to revisit the above estimates and to get the strict inequality $E(\tilde{\mathcal{F}}) > E(\mathcal{F}_\mu^*)$ since the perturbation is nontrivial.
\section*{Acknowledgements} 
\olive This work was supported by  the Deutsche Forschungsgemeinschaft (DFG, German Research Foundation) under Germany's Excellence Strategy EXC 2044 -390685587, Mathematics M\"unster: Dynamics--Geometry--Structure. E.M. acknowledges support from the MIUR-PRIN  project  No 2017TEXA3H and from the  INDAM-GNAMPA 2019 project   {\it ``Trasporto ottimo per dinamiche con interazione''}. P.P. is partially supported by the Austrian Science Fund (FWF)
project P~29681 and by the Vienna Science and Technology Fund (WWTF),
the city of Vienna, and Berndorf Privatstiftung through Project
MA16-005. \EEE

The \PPP authors  \EEE would like to \ZZZ thank \EEE Marco Morandotti and Davide Zucco, organizers of the conference {\it Analysis and applications: contribution from young researchers}, held at the Politecnico di Torino on 8-9 April 2019, where the \ZZZ content \EEE of this paper was presented.



\bibliographystyle{alpha}

\end{document}